\documentclass[conference]{IEEEtran}
\ifCLASSINFOpdf
  \usepackage[pdftex]{graphicx}

\else

\fi

\usepackage{amsmath, amssymb, amsfonts, cite, lipsum, amsthm, mathrsfs, subcaption, balance, xcolor, float}
\newtheorem{theorem}{Theorem}
\newtheorem{proposition}{Proposition}
\let\bs\boldsymbol
\DeclareMathOperator*{\argmax}{\arg\!\max}
\begin{document}
\title{On the Performance of NOMA-based Cooperative Relaying with Receive Diversity}

\author{\IEEEauthorblockN{Vaibhav Kumar,
Barry Cardiff, and Mark F. Flanagan}
\IEEEauthorblockA{School of Electrical and Electronic Engineering, 
University College Dublin, Belfield, Dublin 4, Ireland\\
Email: vaibhav.kumar@ucdconnect.ie, barry.cardiff@ucd.ie, mark.flanagan@ieee.org}}
\maketitle

\begin{abstract}
Non-orthogonal multiple access (NOMA) is widely recognized as a potential multiple access (MA) technology for efficient spectrum utilization in the fifth-generation (5G) wireless standard. In this paper, we present the achievable sum rate analysis of a cooperative relaying system (CRS) using NOMA with two different receive diversity schemes -- selection combining (SC), where the antenna with highest instantaneous signal-to-noise ratio (SNR) is selected, and maximal-ratio combining (MRC). We also present the outage probability and diversity analysis for the CRS-NOMA system. Analytical results confirm that the CRS-NOMA system outperforms the CRS with conventional orthogonal multiple access (OMA) by achieving higher spectral efficiency at high transmit SNR and achieves a full diversity order.
\end{abstract}

\IEEEpeerreviewmaketitle

\section{Introduction}
NOMA has recently been recognized as a promising multiple access technology for 5G wireless networks and beyond, as it can meet the ubiquitous and heterogeneous demands on low latency and high reliability, and can support massive connectivity by providing high throughput and better spectral efficiency~\cite{Bhargava}. It enables multiple users to simultaneously share a time slot, a frequency channel and/or a spreading code, via multiplexing them in the power domain at the transmitter and using successive interference cancellation (SIC) at the receiver to remove messages intended for other users. 

An interesting application of NOMA for a power-domain multiplexed system using cooperative relaying in Rayleigh distributed block fading channels was proposed in~\cite{CRS_NOMA}, where the source was able to deliver two data symbols to the destination in two time slots with the help of a relay.  The advantage of such a system can easily be seen in terms of throughput, compared to the conventional OMA relaying system where a single symbol is delivered to the destination in two time slots. In particular, closed-form expressions for the average achievable sum rate and for near-optimal power allocation were derived in~\cite{CRS_NOMA}. A performance analysis of the CRS-NOMA system over Rician fading channels was presented in~\cite{CRS_NOMA_Rician}, where the authors developed an analytical framework for the average achievable sum-rate and also proposed a method to calculate the approximate achievable rate by using Gauss-Chebyshev integration.

In this paper, we investigate the performance of the CRS-NOMA system for the case when the relay and the destination are equipped with multiple receive antennas. We consider two different diversity combining techniques at the relay and destination receivers -- SC and MRC. We derive closed-form expressions for the average achievable sum-rate and outage probability for the CRS-NOMA system for the cases of SC and MRC receivers. For the purpose of comparison, we present numerical results for the achievable rate of the CRS-OMA system with SC and MRC. In order to have a better insight into the system performance, we present the diversity analysis for the CRS-NOMA system and prove analytically that the system achieves full diversity order for both SC and MRC schemes.

\section{System Model}
Consider the CRS-NOMA model shown in Fig. \ref{SysMod}, which consists of a source $S$ with a single transmit antenna, a relay $R$ with $N_{r}$ receive antennas and a single transmit antenna, and a destination $D$ with $N_{d}$ receive antennas. All nodes are assumed to be operating in half-duplex mode and all wireless links are assumed to be independent and Rayleigh distributed. The channel coefficient between the source and the $i^{\text{th}}$ relay antenna $(1 \le i \le N_r)$ is denoted by $h_{s r, i}$ and has mean-square value $\Omega_{sr}$ for any value of $i$, while that between the source and the $j^{\text{th}}$ destination antenna $(1 \le j \le N_d)$ is denoted by $h_{sd, j}$ and has the mean-square value $\Omega_{sd}$ for any value of $j$. Similarly, the channel coefficient between the relay and the $k^{\text{th}}$ destination antenna $(1 \le k \le N_d)$ is denoted by $h_{r d, k}$ and has mean-square value $\Omega_{rd}$ for any value of $k$. Furthermore, it is assumed that the channels between the source and the destination are on average weaker than those between the source and the relay, i.e., $\Omega_{sd} < \Omega_{sr}$.
\vskip-0.15in
\begin{figure}[hbht]
\centering
\includegraphics[scale=0.7]{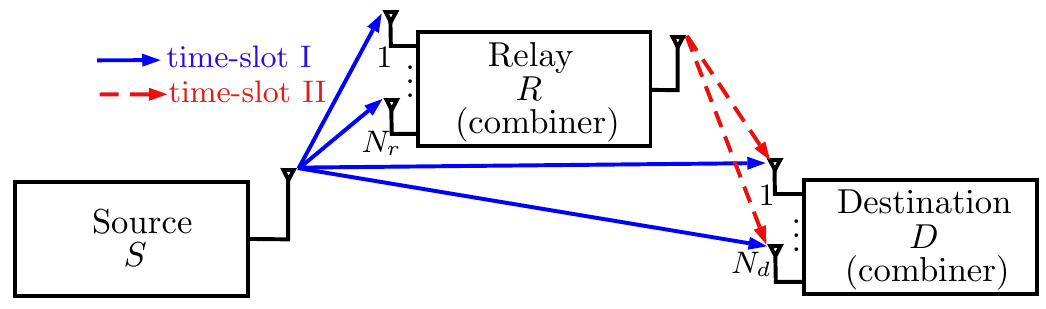}
\caption{System model for CRS-NOMA with multiple receive antennas.}
\label{SysMod}.
\end{figure} 
\vskip-0.2in
In the CRS-NOMA scheme, the source broadcasts $\sqrt{a_{1}P_{t}} s_{1} + \sqrt{a_{2}P_{t}} s_{2}$ to both relay and destination, where $s_1$ and $s_2$ are the data-bearing constellation symbols which are multiplexed in the power domain $(\mathbb{E}\{ | s_i |^2 \} = 1 \text{ for } i=1,2)$, $P_t$ is the the total power transmitted from the source, and $a_1$ and $a_2$ are power weighting coefficients satisfying the constraints $a_1 + a_2 = 1$ and $a_1 > a_2$. Upon reception, the destination decodes symbol $s_{1}$ treating interference from $s_2$ as additional noise, while the relay first decodes symbol $s_{1}$ and then applies SIC to decode symbol $s_{2}$. In the second time slot, the source remains silent and only the relay transmits its estimate of symbol $s_2$, denoted by $\hat{s}_2$, to the destination with full transmit power $P_{t}$. In this manner, two different symbols are delivered to the destination in two time slots. 

In contrast to this, in the conventional OMA scheme, the source broadcasts symbol $s_{1}$ with power $P_{t}$ to both relay and destination in the first time slot and the relay retransmits the estimate of symbol $s_1$, denoted by $\hat{s}_{1}$, to the destination in the second time slot. The destination then combines both copies of symbol $s_{1}$ and in this manner only a single symbol is delivered to the destination in two time slots. 

\section{Performance Analysis} 
In this section, we present the achievable sum-rate, outage probability and diversity analysis of the CRS-NOMA system with two different receive diversity combining techniques, namely SC and MRC.

\subsection{Reception using SC for CRS-NOMA}
The signal received at the relay (resp. destination) in the first time slot is given by
\begin{align*}
	y_{s\mu, \mathrm{SC}} & = h_{s\mu, i^*} \left(\sqrt{a_{1}P_{t}} s_{1} + \sqrt{a_{2}P_{t}} s_{2}\right) + n_{s\mu}, 
\end{align*}
where $\mu = r$ (resp. $\mu = d$) and $i^* = \argmax_{1 \leq i \leq N_{\mu}}(|h_{s\mu, i}|)$. Moreover, $n_{s\mu}$ denotes complex additive white Gaussian noise (AWGN) with zero mean and variance $\sigma^{2}$. The received instantaneous signal-to-interference-plus-noise ratio (SINR) at the relay for decoding symbol $s_{1}$ and the instantaneous signal-to-noise ratio (SNR) for decoding symbol $s_{2}$ (assuming the symbol $s_{1}$ is decoded correctly) are $\gamma_{sr, \mathrm{SC}}^{(1)} = \frac{\delta_{sr}a_{1}P_{t}}{\delta_{sr}a_{2}P_{t} + \sigma^{2}}$ and  $\gamma_{sr, \mathrm{SC}}^{(2)} = \frac{\delta_{sr}a_{2}P_{t}}{\sigma^{2}}$, respectively, where $\delta_{sr} = |h_{sr, i^*}|^{2}$. Similarly, the received instantaneous SINR at the destination for the decoding of symbol $s_{1}$ is given by $\gamma_{sd, \mathrm{SC}} = \frac{\delta_{sd}a_{1}P_{t}}{\delta_{sd}a_{2}P_{t} + \sigma^{2}}$, where $\delta_{sd} = |h_{sd, j^*}|^{2}$. In the next time slot, the relay transmits the decoded symbol $\hat{s}_{2}$ to the destination with power $P_{t}$. The received signal at the destination is given by 
\begin{equation}
	y_{rd, \mathrm{SC}} = h_{rd, k^*} \sqrt{P_{t}}\hat{s}_{2} + n_{rd}, \notag 
\end{equation}
where $k^* = \argmax_{1 \leq k \leq N_{d}} (|h_{rd, k}|)$ and $n_{rd}$ is zero-mean complex AWGN with variance $\sigma^2$. The received instantaneous SNR at the destination while decoding the symbol $s_{2}$ is given by $\gamma_{rd, \mathrm{SC}} = \frac{\delta_{rd}P_{t}}{\sigma^{2}}$, where $\delta_{rd} = |h_{rd, k^*}|^{2}$. Since the symbol $s_{1}$ should be correctly decoded at the destination as well as at the relay for SIC, the average achievable rate for the symbol $s_{1}$ is given by (c.f.~\cite{CRS_NOMA_Rician})
\begin{align}
	\bar{C}_{s_{1}, \!\mathrm{SC}} = & \, \dfrac{1}{2 \ln(2)} \left[ \!\rho\!\! \int_{0}^{\infty} \!\!\dfrac{1\! -\! F_{X}(x)}{1 + \rho x} dx\! -\! \rho a_2 \!\!\!\int_{0}^{\infty} \!\!\dfrac{1 - F_{X}(x)}{1 + \rho a_2 x} \, dx\right] \notag \\
	= & \dfrac{1}{2\ln(2)}(I_{1} - I_{2}), \label{C_s1_SC_integral}
\end{align}
where $\rho = P_{t}/\sigma^{2}$ is the transmit SNR, $X \triangleq \min \{ \delta_{sr}, \delta_{sd}\}$ and $F_{X}(x)$ denotes the cumulative distribution function (CDF) of the random variable $X$.
\begin{theorem}
A closed-form expression for the average achievable rate for symbol $s_1$ in Rayleigh fading using SC in CRS-NOMA  is given by
\begin{align}
	& \bar{C}_{s_1, \mathrm{SC}} = \dfrac{1}{2 \ln(2)} \sum_{k = 1}^{N_r} \sum_{j = 1}^{N_d} (-1)^{k + j} \binom{N_r}{k} \binom{N_d}{j} \notag \\ 
	 &  \hspace{0.3cm} \times \left[ \!\exp \! \left( \! \dfrac{\chi_{k, j}}{\rho} \!\right)\! \Gamma \! \left( \! 0, \dfrac{\chi_{k, j}}{\rho} \!\right) \! -\! \exp \left( \! \dfrac{\chi_{k, j}}{\rho a_2} \!\right) \!\Gamma\! \left( \!0, \dfrac{\chi_{k, j}}{\rho a_2}\!\right) \!\right], \label{C_s1_SC_Closed}
\end{align} 
where $\chi_{k, j} = (k/\Omega_{sr}) + (j/\Omega_{sd})$ and $\Gamma(\cdot, \cdot)$ denotes the upper-incomplete Gamma function.
\end{theorem}

\noindent \emph{Proof}: See Appendix A. 

The average achievable rate for symbol $s_2$ is given by~(c.f.~\cite{CRS_NOMA_Rician})
\begin{align}
	\bar{C}_{s_2, \mathrm{SC}} = \dfrac{\rho}{2 \ln(2)} \int_{0}^{\infty}\dfrac{1 - F_{Y}(x)}{1 + \rho x}\, dx, \label{C_s2_SC_integral}
\end{align}
where $Y \triangleq \min\{\delta_{sr}a_2, \delta_{rd}\}$.
\begin{proposition}
A closed-form expression for the average achievable rate for symbol $s_2$ in Rayleigh fading using SC in CRS-NOMA is given by
\begin{align}
	\!\!\!\!\bar{C}_{s_2, \mathrm{SC}} \!= \!&  \sum_{k = 1}^{N_r} \! \sum_{j = 1}^{N_d} \!\!\dfrac{(-1)^{k + j}}{2\ln(2)}   \binom{N_r}{k} \!\! \binom{N_d}{j} \!\! \exp\!\! \left(\! \dfrac{\theta_{k, j}}{\rho}\!\!\right) \!\!\Gamma \!\!\left(\! 0, \!\dfrac{\theta_{k, j}}{\rho}\!\!\right), \label{C_s2_SC_Closed}
\end{align}
where $\theta_{k, j} = \tfrac{k}{\Omega_{sr}a_2} + \tfrac{j}{\Omega_{rd}}$. 
\end{proposition}
\begin{proof}
	Analogous to the arguments in Appendix A and using a transformation of random variables, we have 
	\begin{align}
		1 - F_Y(x) = \sum_{k = 1}^{N_r} \!\sum_{j = 1}^{N_d} \!(-1)^{k + j}\! \binom{N_r}{k} \!\binom{N_d}{j}\! \exp \!\left( -\theta_{k, j} x\right), \label{1-FY}
\end{align}
Substituting the expression of $1 - F_{Y}(x)$ from \eqref{1-FY} into \eqref{C_s2_SC_integral} and solving the integration using~\cite[eqn.~(3.352-4),~p.~341]{Grad}, the closed-form expression for $\bar{C}_{s_2, \mathrm{SC}}$ reduces to~\eqref{C_s2_SC_Closed}.
\end{proof}

The average achievable sum-rate for the CRS-NOMA system using SC in Rayleigh fading is therefore given using \eqref{C_s1_SC_Closed} and \eqref{C_s2_SC_Closed} as 
\begin{align}
	\bar{C}_{\mathrm{sum, SC}} = \bar{C}_{s_1, \mathrm{SC}} + \bar{C}_{s_2, \mathrm{SC}}. \label{C_Sum_SC} 
\end{align}
It is interesting to note that for $N_r = N_d = 1$, \eqref{C_Sum_SC} reduces to~\cite[eqn.~(14)]{CRS_NOMA}. 

\subsection{Reception using SC for CRS-OMA}
The signal received at the relay (resp. destination) in the first time slot is given by 
\begin{align}
	y_{s\mu, \mathrm{SC-OMA}} & = h_{s\mu, i^*}\, \sqrt{P_t} \, s_{1} + n_{s\mu}, \notag 
\end{align}
where $\mu = r$ (resp. $\mu = d$) and $i^* = \argmax_{1 \leq i \leq N_{\mu}}(|h_{s\mu, i}|)$. In the next time slot, the relay forwards its estimate of $s_1$, denoted by $\hat{s}_1$, to the destination. The signal received at the destination is given by 
\begin{align}
	y_{rd, \mathrm{SC-OMA}} & =  h_{rd, k^*} \, \sqrt{P_{t}} \, \hat{s}_{1} + n_{rd}. \notag 
\end{align}

The average achievable rate for the symbol $s_{1}$ is given by~(c.f.~\cite{Laneman})
\begin{align}
	\bar{C}_{\mathrm{SC-OMA}} & = 0.5\, \mathbb E_{W} \left[\log_{2} (1 + W\rho)\right],  \label{C_SC-OMA_integral}
\end{align}
where\footnote{The rate calculation is based on the assumption that the destination performs SC in the first and the second time slots and then applies MRC on the resulting signals from the two time slots. In the case where the destination applies SC on the resulting signals instead of MRC, $W$ will instead be defined as $\min \{ \delta_{sr}, \max\{\delta_{sd}, \delta_{rd}\}\}$, and this will result in a performance degradation with respect to the system described here.} $W \triangleq \min\{\delta_{sr}, \delta_{sd} + \delta_{rd}\}$ and $\mathbb E_{\mathscr Z}[\cdot]$ denotes the expectation with respect to the random variable $\mathscr Z$. Since the focus of this paper is on the NOMA-based systems, we do not present a closed-form analysis for CRS-OMA.

\subsection{Outage probability for CRS-NOMA using SC}
In this subsection, we will characterize the outage probability of symbols $s_{1}$ and $s_{2}$ for the CRS-NOMA using selection combining in Rayleigh fading. We define $\mathcal{O}_{1, \mathrm{SC}}$ as the outage event for symbol $s_1$ using SC, i.e., the event where either the relay or the destination fails to decode $s_{1}$ successfully. Hence the outage probability for the symbol $s_{1}$ is given by 
\begin{align}
	& \Pr(\mathcal{O}_{1, \mathrm{SC}}) = \Pr(C_{s_1, \mathrm{SC}} < R_{1}) \notag \\
	= & \Pr \left[ \dfrac{1}{2} \log_{2} \left( 1 + \dfrac{a_{1} \rho X}{1 + a_{2}\rho X}\right) < R_{1}\right] = \Pr(X < \Theta_{1}) \notag \\
	= &  F_{\delta_{sr}}(\Theta_{1}) + F_{\delta_{sd}}(\Theta_{1}) - F_{\delta_{sr}}(\Theta_{1}) F_{\delta_{sd}}(\Theta_{1}), \label{P_out_s1_SC}
\end{align}
where $C_{s_1, \mathrm{SC}}$ is the instantaneous achievable rate of symbol $s_1$ in CRS-NOMA using SC in Rayleigh fading, $R_{1}$ is the target data rate for the symbol $s_{1}$, $\epsilon_{1} = 2^{2R_{1}}~-~1$ and $\Theta_{1} = \tfrac{\epsilon_{1}}{\rho (a_{1} - \epsilon_{1}a_{2})}$. The system design must ensure that $a_{1} > \epsilon_{1}a_{2}$, otherwise the outage probability for symbol $s_{1}$ will always be~1 as noted in \cite{RelaySelectionDing}. The closed-form expressions for $F_{\delta_{sr}}(\Theta_{1})$ and $F_{\delta_{sd}}(\Theta_{1})$ are given in~Appendix~A. Next, we define $\mathcal{O}_{2, \mathrm{SC}}$ as the outage event for symbol $s_2$ using SC. This outage event can be decomposed as the union of the following disjoint events: (i)~symbol $s_1$ cannot be successfully decoded at the relay; (ii)~symbol $s_1$ is successfully decoded at the relay, but symbol $s_2$ cannot be successfully decoded at the relay; and (iii)~both symbols are successfully decoded at the relay, but symbol $s_2$ cannot be successfully decoded at the destination. Therefore, the outage probability for the symbol $s_2$ may be expressed as
\begin{align}
	\Pr(\mathcal{O}_{2, \mathrm{SC}}) & = \begin{cases} \Pr(\delta_{sr} < \Theta_{1}) + \Pr(\delta_{sr} \geq \Theta_{1}, \delta_{sr} < \Theta_{2}) & \\
	\hspace{0.1cm}+ \Pr (\delta_{sr} > \Theta_{2}, \delta_{rd} < \epsilon_{2}/\rho); \,\operatorname{if}\,\,\Theta_{1} < \Theta_{2} & \\
	\Pr(\delta_{sr} < \Theta_{1}) + \Pr(\delta_{sr} > \Theta_{1}, \delta_{rd} < \epsilon_{2}/\rho); & \\
	\hspace{4.5cm} \operatorname{otherwise}\end{cases} \notag \\
	& = F_{\delta_{sr}} (\Theta) \!+\! F_{\delta_{rd}} (\epsilon_{2}/\rho) \!-\! F_{\delta_{sr}}(\Theta) F_{\delta_{rd}}(\epsilon_{2}/\rho),\!\!\! \label{P_out_s2_SC}
\end{align}
where $R_{2}$ is the target data rate for the symbol $s_{2}$, $\epsilon_{2}~=~2^{2R_{2}}~-~1$, $\Theta_{2} = \tfrac{\epsilon_{2}}{a_{2}\rho}$ and $\Theta = \max\{\Theta_{1}, \Theta_{2}\}$. The closed-form expressions for $F_{\delta_{sr}}(\Theta)$ and $F_{\delta_{rd}}(\epsilon_{2}/\rho)$ are given in~Appendix~A.

\subsection{Diversity analysis for CRS-NOMA using SC}
From \eqref{FdeltaSR}, we have 
\begin{align}
	& F_{\delta_{sr}}(\Theta_{1}) = \sum_{k = 1}^{N_r} (-1)^{k - 1} \binom{N_{r}}{k} \left[ 1 - \exp \left( \dfrac{-k\Theta_{1}}{\Omega_{sr}}\right) \right]  \notag \\
	& = \sum_{k = 1}^{N_r} \sum_{l = 1}^{\infty} \dfrac{(-1)^{k + l}}{l!} \binom{N_r}{k} \left( \dfrac{k\Theta_{1}}{\Omega_{sr}}\right)^{l} = \sum_{l = N_r}^{\infty} \dfrac{(-1)^{l}\Theta_{1}^{l}}{l!\Omega_{sr}^{l}}  \notag \\
	& \hspace{0.5cm}\times \sum_{k = 1}^{N_{r}} \binom{N_r}{k} (-1)^{k}k^{l} \tag{Using \cite[eqn.~(0.154-3), p. 4]{Grad}} \\
	& = \dfrac{(-1)^{N_r} \epsilon_{1}^{N_r}}{(a_{1} - \epsilon_{1}a_{2})^{N_r} N_r ! \Omega_{sr}^{N_r}}  \notag \\
	&\hspace{1cm} \times \sum_{k = 1}^{N_r}\binom{N_r}{k} (-1)^{k}k^{N_r}\rho^{-N_r}+ \mathbb{O}\left[\rho^{-(N_r + 1)}\right], \label{F_deltaSR_order}
\end{align}
where $\mathbb{O}$ is the Landau symbol. Hence it is clear from \eqref{F_deltaSR_order} that $F_{\delta_{sr}}(\Theta_{1})$ decays as $\rho^{-N_r}$ as $\rho \to \infty$. Similarly, it can be easily shown that $F_{\delta_{sd}}(\Theta_{1})$ decays as $\rho^{-N_d}$ and $F_{\delta_{sr}}(\Theta_{1}) F_{\delta_{sd}}(\Theta_{1})$ decays as $\rho^{-(Nr + N_d)}$ as $\rho \to \infty$. Therefore it is straightforward to conclude using \eqref{P_out_s1_SC} that the diversity order of the symbol $s_{1}$ is $\min\{N_r, N_d, N_r N_d\} = \min\{N_r, N_d\}$. Following similar arguments, it can be shown that the diversity order of the symbol $s_{2}$ is $\min\{N_r, N_d\}$.

\subsection{Reception using MRC for CRS-NOMA}
The signal received at the relay (resp. destination) in the first time slot is given by
\begin{align}
	& y_{s\mu, \mathrm{MRC}} = \bs{h}_{s\mu}^{H} \, \left( \bs{h}_{s\mu} \left(\sqrt{a_{1}P_{t}} s_{1} + \sqrt{a_{2}P_{t}}s_{2}\right) + \bs{n}_{s\mu}\right), \notag  
\end{align}
where $\mu = r$ (resp. $\mu = d$), $\bs{h}_{s\mu} = [h_{s\mu, 1}\, h_{s\mu, 2}\, \cdots \, h_{s\mu, N_{\mu}}]^{T} \in \mathbb{C}^{N_{\mu} \times 1}$, $\bs{n}_{s\mu} = [n_{s\mu, 1}\, n_{s\mu, 2}$ $\, \cdots \, n_{s\mu, N_{\mu}}]^{T} \in~\mathbb{C}^{N_{\mu} \times 1}$, $(\cdot)^{H}$ is the Hermitian operator and $(\cdot)^T$ is the transpose operator. The elements in the vector $\bs{h}_{s\mu}$ are independent and distributed as $\mathcal{CN}(0, \Omega_{s\mu})$ and the elements in $\bs{n}_{s\mu}$ are independent and distributed according to $\mathcal{CN}(0, \sigma^2)$.

The received instantaneous SINR at the relay for decoding symbol $s_{1}$ and instantaneous SNR for decoding symbol $s_{2}$ (assuming the symbol $s_{1}$ is decoded correctly) are obtained as $\gamma_{sr, \mathrm{MRC}}^{(1)} = \frac{\lambda_{sr}a_{1}P_{t}}{\lambda_{sr}a_{2}P_{t} + \sigma^{2}}$ and $\gamma_{sr, \mathrm{MRC}}^{(2)} = \frac{\lambda_{sr}a_{2}P_{t}}{\sigma^{2}}$, respectively, where $\lambda_{sr} = \sum_{i = 1}^{N_r}|h_{sr, i}|^{2}$. Similarly, the received instantaneous SINR at the destination while decoding $s_{1}$ is given by $\gamma_{sd, \mathrm{MRC}} = \frac{\lambda_{sd} a_{1}P_t}{\lambda_{sd}a_{2}P_t + \sigma^{2}}$, where $\lambda_{sd} = \sum_{i = 1}^{N_d}|h_{sd, i}|^{2}$. In the next time slot, the relay transmits the decoded symbol $\hat{s}_{2}$ to the destination with power $P_{t}$. The received signal at the destination (after applying MRC) is given by 
\begin{align}
	y_{rd, \mathrm{MRC}} = \bs{h}_{rd}^H \left(  \bs{h}_{rd} \sqrt{P_t}\hat{s}_{2} + \bs{n}_{rd} \right), \notag 
\end{align}
where $\bs{h}_{rd} = \left[h_{rd, 1} \, h_{rd, 2}\, \cdots\, h_{rd, N_d}\right]^T \in \mathbb C^{N_d \times 1}$ with independent elements each distributed as $\mathcal{CN}(0, \Omega_{rd})$ and $\bs{n}_{rd} = [n_{rd, 1}\, n_{rd, 2}$ $\, \cdots \, n_{rd, N_{d}}]^{T} \in~\mathbb{C}^{N_{d} \times 1}$ with independent elements each distributed according to $\mathcal{CN}(0, \sigma^2)$. The received instantaneous SNR at the destination while decoding the symbol $s_{2}$ is $\gamma_{rd, \mathrm{MRC}} = \frac{\lambda_{rd}P_t}{\sigma^{2}}$, where $\lambda_{rd} = \sum_{i = 1}^{N_d}|h_{rd, i}|^2$. The average achievable rate for the symbol $s_{1}$ is given by~(c.f.~\cite{CRS_NOMA_Rician})
\begin{align}
	\bar{C}_{s_1, \mathrm{MRC}} \!\!& = \!\!\dfrac{1}{2 \ln (2)} \!\left[\!\rho\!\!\int_{0}^{\infty}\!\!\dfrac{1 \!-\! F_{\mathcal{X}}(x)}{1 + x \rho}dx \!-\! \rho a_2 \!\!\int_{0}^{\infty}\!\! \dfrac{1 \!-\! F_{\mathcal{X}}(x)}{1 + x \rho a_{2}}dx \right]\notag \\
	& = \dfrac{1}{2 \ln (2)} (I_3 - I_4), \label{C_s1_MRC_Integral}
\end{align}
where $\mathcal{X} = \min\{\lambda_{sr}, \lambda_{sd}\}$.
\begin{theorem}
	A closed-form expression for the average achievable rate for symbol $s_1$ for CRS-NOMA using MRC in Rayleigh fading is given by
	\begin{align}
		& \bar{C}_{s_1, \mathrm{MRC}} = \dfrac{1}{2 \ln(2)} \sum_{i = 0}^{N_r-1} \sum_{j = 0}^{N_d-1} \dfrac{\Gamma(1 + i + j)}{i! j! \Omega_{sr}^i \Omega_{sd}^j \rho^{i + j}}\Bigg[ \exp\left( \dfrac{\phi}{\rho}\right) \notag \\
		& \times \!\left. \Gamma \!\left( -i-j, \dfrac{\phi}{\rho}\right) \!-\! \dfrac{1}{a_2^{i + j}} \exp \!\left( \dfrac{\phi}{\rho a_2}\right)\! \Gamma\! \left( \!-i\!-j, \!\!\dfrac{\phi}{\rho a_2}\right)\!\right], \label{C_s1_MRC_Closed}
	\end{align}
where $\phi = \Omega_{sr}^{-1} + \Omega_{sd}^{-1}$ and $\Gamma(\cdot)$ denotes the Gamma function.
\end{theorem}

\noindent \emph{Proof}: See Appendix B.

The average achievable rate for the symbol $s_{2}$ is given by~(c.f.~\cite{CRS_NOMA_Rician})
\begin{align}
	\bar{C}_{s_{2}, \mathrm{MRC}} = \dfrac{\rho}{2 \ln (2)} \int_{0}^{\infty} \dfrac{1 - F_{\mathcal{Y}} (x)}{1 + x \rho}dx, \label{C_s2_MRC_Integral}
\end{align}
where $\mathcal{Y} \triangleq \min\{\lambda_{sr}a_{2}, \lambda_{rd}\}$.
\begin{proposition}
	The closed-form expression for the average achievable rate for symbol $s_2$ for CRS-NOMA using MRC in Rayleigh fading is obtained as
	\begin{align}
		\bar{C}_{s_{2}, \mathrm{MRC}} & = \dfrac{1}{2 \ln(2)} \sum_{i = 0}^{N_{r} - 1} \sum_{j = 0}^{N_{d} - 1} \dfrac{\Gamma(1 + i + j) }{i! \,j!\, a_{2}^{i}\,\Omega_{sr}^{i} \, \Omega_{rd}^{j} \rho^{(i + j)}} \notag \\
	& \hspace{2cm}\times    \exp \left( \dfrac{ \xi}{\rho}\right)\Gamma\left(-i - j, \dfrac{\xi}{\rho}\right), \label{C_s2_MRC_Closed}
	\end{align}
where $\xi = (\Omega_{sr} a_2)^{-1} + \Omega_{rd}^{-1}$.	
\end{proposition}
\begin{proof}
Similar to the arguments in Appendix B and using a transformation of random variables, we have,
\begin{align}
	1 - F_{\mathcal{Y}}(x) = & \exp (-x \xi) \sum_{i = 0}^{N_r - 1} \sum_{j = 0}^{N_d - 1} \dfrac{x^{i + j}}{i! j! a_{2}^{i} \Omega_{sr}^{i} \Omega_{rd}^{j}}. \label{OneMinusF_mathcalY}
\end{align}
Substituting $1 - F_{\mathcal{Y}}(x)$ from \eqref{OneMinusF_mathcalY} into \eqref{C_s2_MRC_Integral} and solving the integral using~\cite[eqn.~(3.383-10), p. 348]{Grad}, the closed-form expression for $\bar{C}_{s_2, \mathrm{MRC}}$ becomes equal to~\eqref{C_s2_MRC_Closed}.
\end{proof}

Hence, the average achievable sum-rate for the CRS-NOMA using MRC in Rayleigh fading is obtained using \eqref{C_s1_MRC_Closed} and \eqref{C_s2_MRC_Closed} as 
\begin{equation}
	\bar{C}_{\mathrm{sum, MRC}} = \bar{C}_{s_{1}, \mathrm{MRC}} + \bar{C}_{s_{2}, \mathrm{MRC}}. \label{C_sum_MRC}
\end{equation}
It is important to note that for $N_r = N_d = 1$, \eqref{C_sum_MRC} reduces to~\cite[eqn.~(14)]{CRS_NOMA}.

\subsection{Reception using MRC for CRS-OMA}
The signals received in the first time slot at the relay (resp. destination) is given by 
\begin{align}
	y_{s\mu, \mathrm{MRC - OMA}} & = \bs{h}_{s\mu}^{H} \left( \bs{h}_{s\mu} \sqrt{P_t}s_1 + \bs{n}_{s\mu}\right), \notag 
\end{align}
where $\mu = r$ (resp. $\mu = d$). In the next time slot, the relay forwards its estimate of $s_1$, denoted by $\hat{s}_1$, to the destination. The signal received at the destination is given by 
\begin{align}
	y_{rd, \mathrm{MRC - OMA}} & = \bs{h}_{rd}^H \left( \bs{h}_{rd} \sqrt{P_t} \hat{s_1} + \bs{n}_{rd} \right). \notag
\end{align}

Similar to the case of CRS-OMA using SC, the average achievable rate for symbol $s_1$ in CRS-OMA using MRC is given by~(c.f.~\cite{Laneman})
\begin{align}
	\bar{C}_{\mathrm{MRC-OMA}} = 0.5\, \mathbb E_{\mathcal Z}\left[\log_2(1 + \mathcal Z \rho) \right], \label{C_MRC-OMA_Integral}
\end{align}
where $\mathcal{Z} \triangleq \min(\lambda_{sr}, \lambda_{sd} + \lambda_{rd})$.

\subsection{Outage probability for CRS-NOMA using MRC}
Similar to the CRS-NOMA using SC, we define $\mathcal{O}_{1, \mathrm{MRC}}$ as the event that the symbol $s_{1}$ is in outage in the CRS-NOMA using MRC in Rayleigh fading. Hence,
\begin{align}
	\Pr(\mathcal{O}_{1, \mathrm{MRC}}) \!& = \Pr(C_{s_{1}, \mathrm{MRC}} < R_{1}) = F_{\mathcal{X}} (\Theta_{1}) \notag \\
	& = \!F_{\lambda_{sr}}\!(\!\Theta_1\!)\! +\! F_{\lambda_{sd}}(\!\Theta_1\!) \!-\! F_{\lambda_{sr}}(\!\Theta_1\!)F_{\lambda_{sd}}(\!\Theta_1\!), \label{P_out_s1_MRC}
\end{align}
where $C_{s_{1}, \mathrm{MRC}}$ is the instantaneous achievable rate for symbol $s_1$ in CRS-NOMA using MRC in Rayleigh fading. Similarly, we define $\mathcal{O}_{2, \mathrm{MRC}}$ as the event that the symbol $s_{2}$ is in outage in the CRS-NOMA using MRC in Rayleigh fading. Therefore, 
\begin{align}
	\Pr(\mathcal{O}_{2, \mathrm{MRC}}) & \!=\!\! F_{\lambda_{sr}}\!(\Theta)\! +\! F_{\lambda_{rd}}\!\left(\!\!\dfrac{\epsilon_{2}}{\rho}\!\!\right)\! -\! F_{\lambda_{sr}}\!(\Theta) F_{\lambda_{rd}}\!\left(\!\!\dfrac{\epsilon_{2}}{\rho}\!\!\right).\label{P_out_s2_MRC}
\end{align}
The closed-form expressions for $F_{\lambda_{sr}}(\Theta_1)$, $F_{\lambda_{sd}}(\Theta_1)$, $F_{\lambda_{sr}}(\Theta)$ and $F_{\lambda_{rd}}(\epsilon_2/\rho)$ can be found using the fact that $\lambda_{sr}$, $\lambda_{sd}$ and $\lambda_{rd}$ are Gamma distributed random variables.

\subsection{Diversity analysis of CRS-NOMA using MRC}
Since $\lambda_{sr}$ is Gamma distributed with shape $N_r$ and scale $\Omega_{sr}$ we have
\begin{equation}
	\begin{aligned}
		F_{\lambda_{sr}}(\Theta_1) = & \dfrac{1}{\Gamma(N_r)}\,\,\gamma\!\left(N_r, \dfrac{\Theta_1}{\Omega_{sr}} \right), 
	\end{aligned}
\end{equation}
where $\gamma(\cdot, \cdot)$ is lower-incomplete Gamma function. Using the series expansion of the lower-incomplete Gamma function as given in~\cite[eqn.~8.11.4,~p.~180]{NIST}, 
\begin{align}
	F_{\lambda_{sr}}\!(\!\Theta_1\!) \!= & \dfrac{1}{\Gamma(N_r)}\!\!\left(\!\! \dfrac{\Theta_1}{\Omega_{sr}}\!\!\right)^{\!\!\!N_r} \!\!\!\!\exp\! \left(\!\! \dfrac{-\Theta_1}{\Omega_{sr}}\!\!\right) \!\!\sum_{k = 0}^{\infty} \dfrac{\Theta_1^k \Gamma(N_r)}{\Omega_{sr}^k \Gamma(N_r + k + 1)}. \notag 
\end{align}
Using the series expansion of the exponential function and replacing $\Theta_1$ by $\tfrac{\epsilon_1}{\rho(a_1 - \epsilon_1 a_2)}$ yields
\begin{align}
	F_{\lambda_{sr}}\!(\!\Theta_1\!) \!= & \sum_{l = 0}^{\infty} \sum_{k = 0}^{\infty}\dfrac{(-1)^l \Theta_1^{N_r + l + k}}{\Omega_{sr}^{N_r + l + k} \Gamma(N_r + l + k)} \notag \\
	= & \dfrac{\epsilon_1^{N_r} \rho^{-N_r}}{(a_1 - \epsilon_1 a_2)^{N_r} \Omega_{sr}^{N_r} \Gamma(N_r)} + \mathbb O\left( \!\rho^{-(N_r + 1)}\!\right). \label{F_lambdaSR_order}
\end{align}
It is clear from \eqref{F_lambdaSR_order} that $F_{\lambda_{sr}}(\Theta_1)$ decays as $\rho^{-N_r}$ as $\rho~\to~\infty$. Similarly, it can be proved that $F_{\lambda_{sd}}(\Theta_1)$ decays as $\rho^{-N_d}$ as $\rho \to \infty$. Also, using the series expansion of the lower-incomplete Gamma function and the exponential function, we have
\begin{align}
	& F_{\lambda_{sr}}\!(\!\Theta_1\!) F_{\lambda_{sd}}\!(\!\Theta_1\!) = \dfrac{1}{\Gamma(\!N_r\!) \Gamma(\!N_d\!)} \gamma \!\left(\!\! N_r, \dfrac{\Theta_1}{\Omega_{sr}}\!\!\right) \gamma \!\left(\!\! N_d, \dfrac{\Theta_1}{\Omega_{sd}}\!\right) \notag \\
	= & \sum_{l = 0}^{\infty}\sum_{k = 0}^{\infty} \sum_{i = 0}^{\infty} \sum_{j = 0}^{\infty} \dfrac{(-1)^{l + i} \Theta_1^{N_r + N_d + l + k + i + j}}{\Omega_{sr}^{N_r + l + k} \Omega_{sd}^{N_d + i + j}} \notag \\
	& \hspace{3cm}\times \dfrac{1}{\Gamma(N_r + l + k) \Gamma(N_d + i + j)} \notag \\
	= & \dfrac{\epsilon_1^{N_r + N_d} \rho^{-(N_r + N_d)}}{(\!a_1 \!-\! \epsilon_1 a_2\!)^{\!N_r + N_d} \Omega_{sr}^{N_r} \Omega_{sd}^{N_d} \Gamma(N_r) \Gamma(N_d)} \!+\! \mathbb O\!\left(\! \rho^{-(N_r + N_d + 1)}\!\right)\!.\label{F_lambdaSRlambdaSD_order}
\end{align}
Hence it is straightforward to conclude using \eqref{P_out_s1_MRC}, \eqref{F_lambdaSR_order} and \eqref{F_lambdaSRlambdaSD_order} that the diversity order for the symbol $s_{1}$ is $\min\{N_r, N_d, N_r N_d\} = \min\{N_r, N_d\}$. Analogously, by representing $F_{\lambda_{sr}}(\Theta)$ and $F_{\lambda_{rd}}(\epsilon_{2}/\rho)$ in \eqref{P_out_s2_MRC} in terms of the lower-incomplete gamma function, it can be shown that the diversity order for the symbol $s_{2}$ is $\min(N_r, N_d, N_r N_d) = \min(N_r, N_d)$.

\section{Results and Discussion}
In this section we present analytical and numerical\footnote{We do not realize the actual scenario for numerical computation, but rather generate the random variables and then evaluate~\eqref{C_Sum_SC},~\eqref{C_SC-OMA_integral},~\eqref{C_sum_MRC} and~\eqref{C_MRC-OMA_Integral}.} results for the average achievable rate and the outage probability for the cooperative relaying system. We consider the CRS system where $\Omega_{sd} = 1$, $\Omega_{sr} = 10$ and $\Omega_{rd} = 2.5$. For all NOMA-based systems, we consider $a_{2} = 0.1$ and $R_{1} = R_2 = 1$~bps/Hz. Fig. \ref{Capacity} shows a comparison of the average achievable rate for the CRS-NOMA (both numerical and analytical results) and CRS-OMA (numerical results) systems. It is clear from the figure that for low transmit SNR $\rho$, the CRS-NOMA system performs worse compared to the conventional CRS-OMA system in terms of achievable rate, but as the transmit SNR $\rho$ becomes large, the CRS-NOMA system outperforms its OMA counterpart for both SC and MRC schemes. It is evident from~Fig.~\ref{CapacitySC} that the CRS-NOMA using SC with $N_r = N_d = 1$ achieves the same spectral efficiency as that of the CRS-OMA using SC with $N_r = N_d = 2$ at high transmit SNR. Also, the CRS-NOMA using SC with $N_r = N_d = 2$ achieves higher spectral efficiency as compared to CRS-OMA using SC with $N_R = N_d = 4$ at high SNR. From~Fig.~\ref{CapacityMRC}, it is clear that the CRS-NOMA using MRC with $N_r = N_d = 2$ achieves the same spectral efficiency as that of the CRS-OMA using MRC with $N_r = N_d = 4$ at high transmit SNR. It can also be noted that the CRS-NOMA system using MRC results in a higher average achievable sum-rate as compared to the CRS-NOMA system using SC. 

Fig. \ref{Outage_SC} shows the outage probability of the symbols $s_{1}$ and $s_{2}$ with varying transmit SNR $\rho$ for the CRS-NOMA system using SC. It is clear that the diversity order for both symbols is $\min(N_r, N_d)$ as derived in Section III-D.

\begin{figure}[H]
\centering
\begin{subfigure}{.23\textwidth}
  \centering
  \includegraphics[width=1\linewidth]{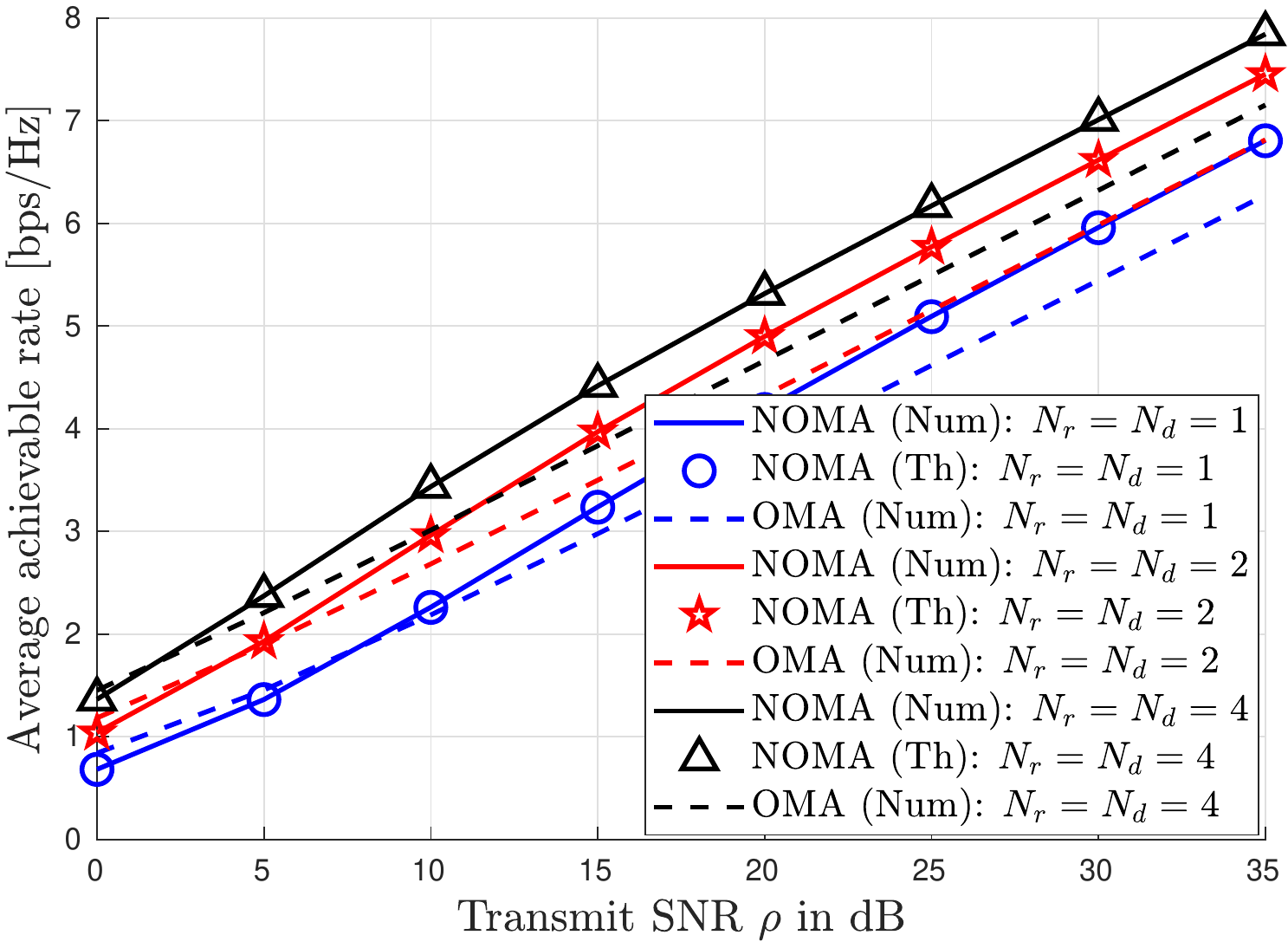}
  \caption{Using SC.}
  \label{CapacitySC}
\end{subfigure}%
\begin{subfigure}{.23\textwidth}
  \centering
  \includegraphics[width=1\linewidth]{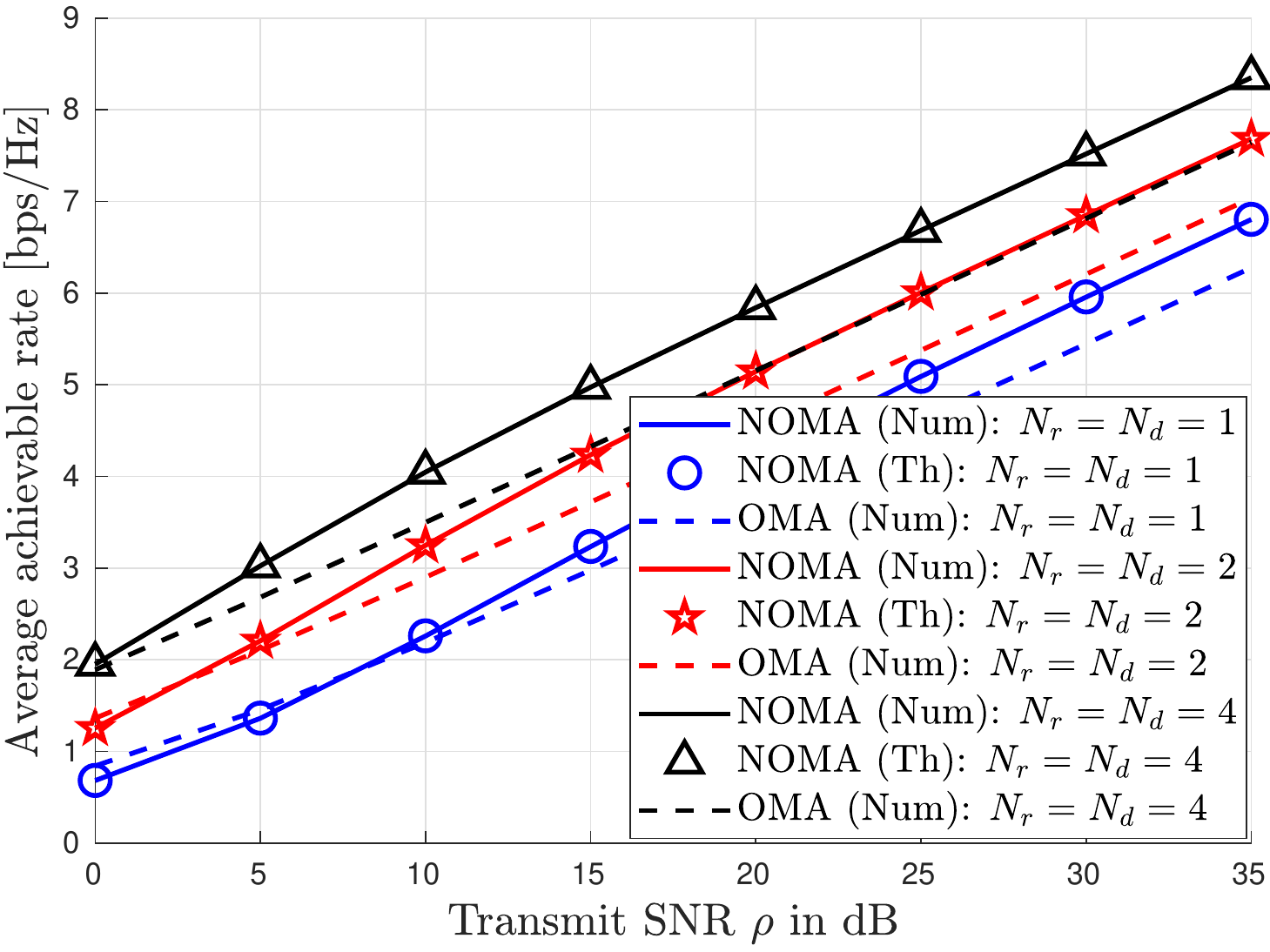}
  \caption{Using MRC.}
  \label{CapacityMRC}
\end{subfigure}
\caption{Average achievable rate for the CRS.}
\label{Capacity}
\end{figure}

\begin{figure}[H]
\centering
\begin{subfigure}{.23\textwidth}
  \centering
  \includegraphics[width=1\linewidth]{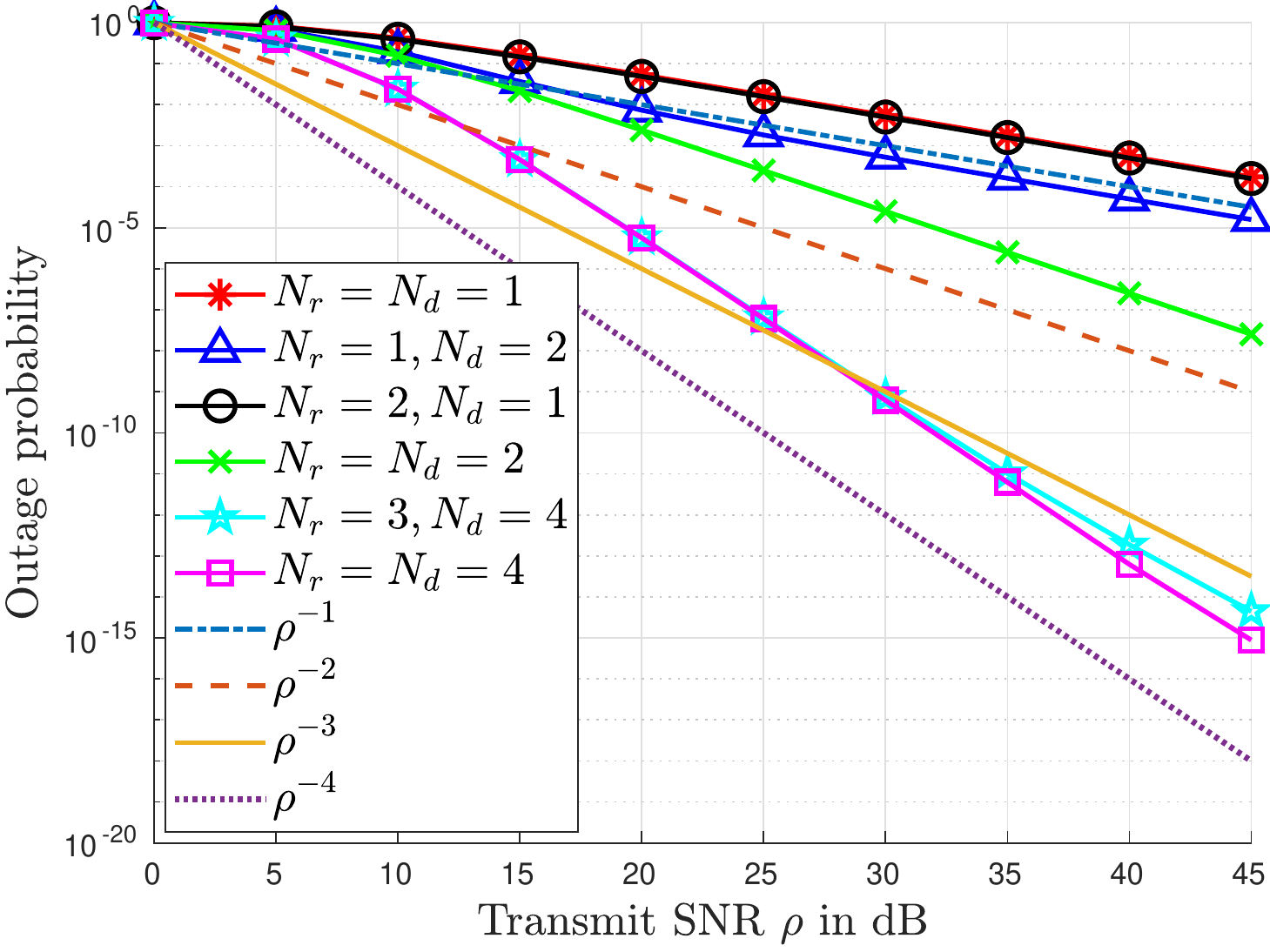}
  \caption{Symbol $s_{1}$}
  \label{Outage_s1_SC}
\end{subfigure}%
\begin{subfigure}{.23\textwidth}
  \centering
  \includegraphics[width=1\linewidth]{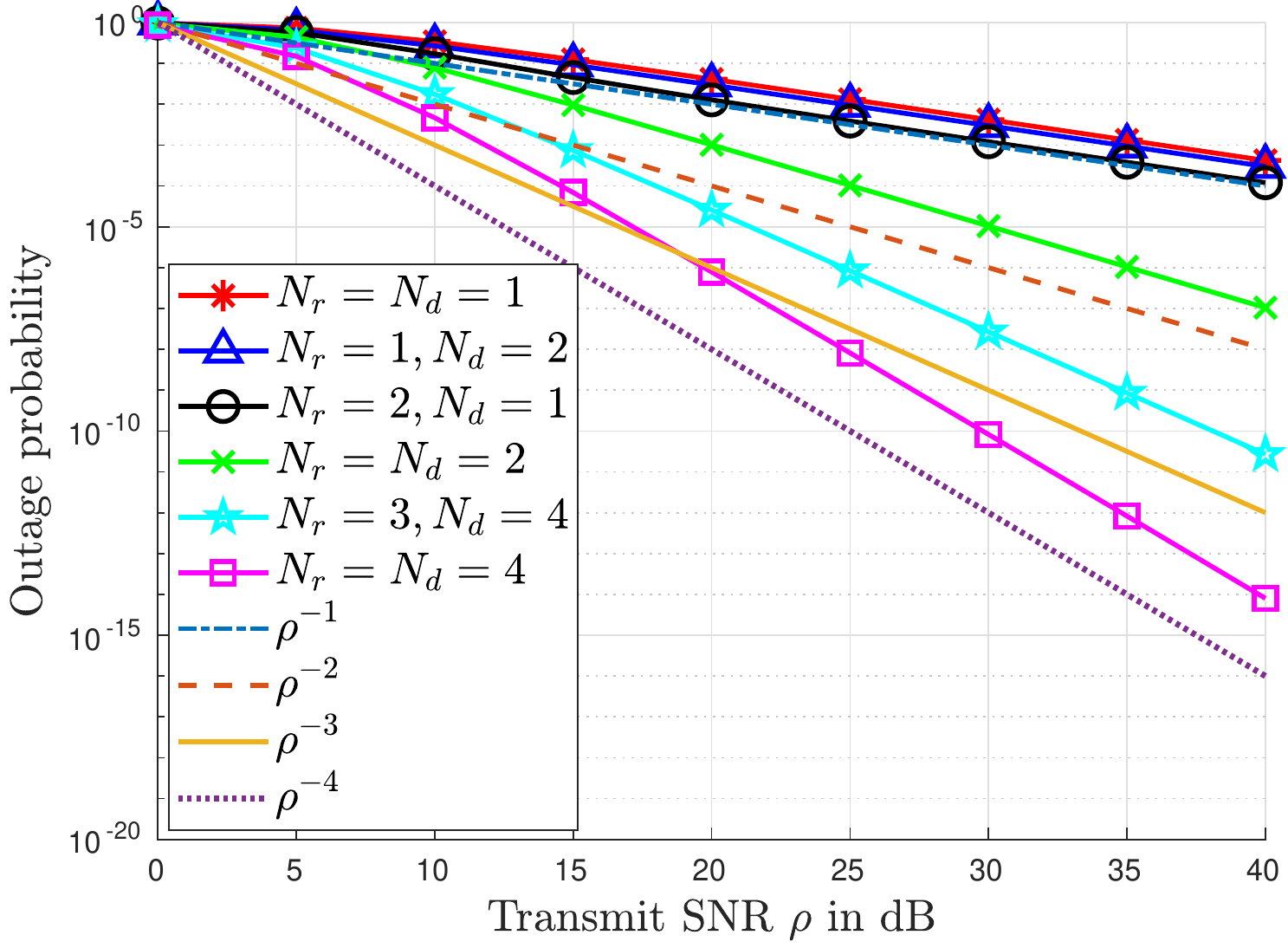}
  \caption{Symbol $s_{2}$}
  \label{Outage_s2_SC}
\end{subfigure}
\caption{Outage probability for CRS-NOMA using SC.}
\label{Outage_SC}
\end{figure}

\begin{figure}[H]
\centering
\begin{subfigure}{.23\textwidth}
  \centering
  \includegraphics[width=1\linewidth]{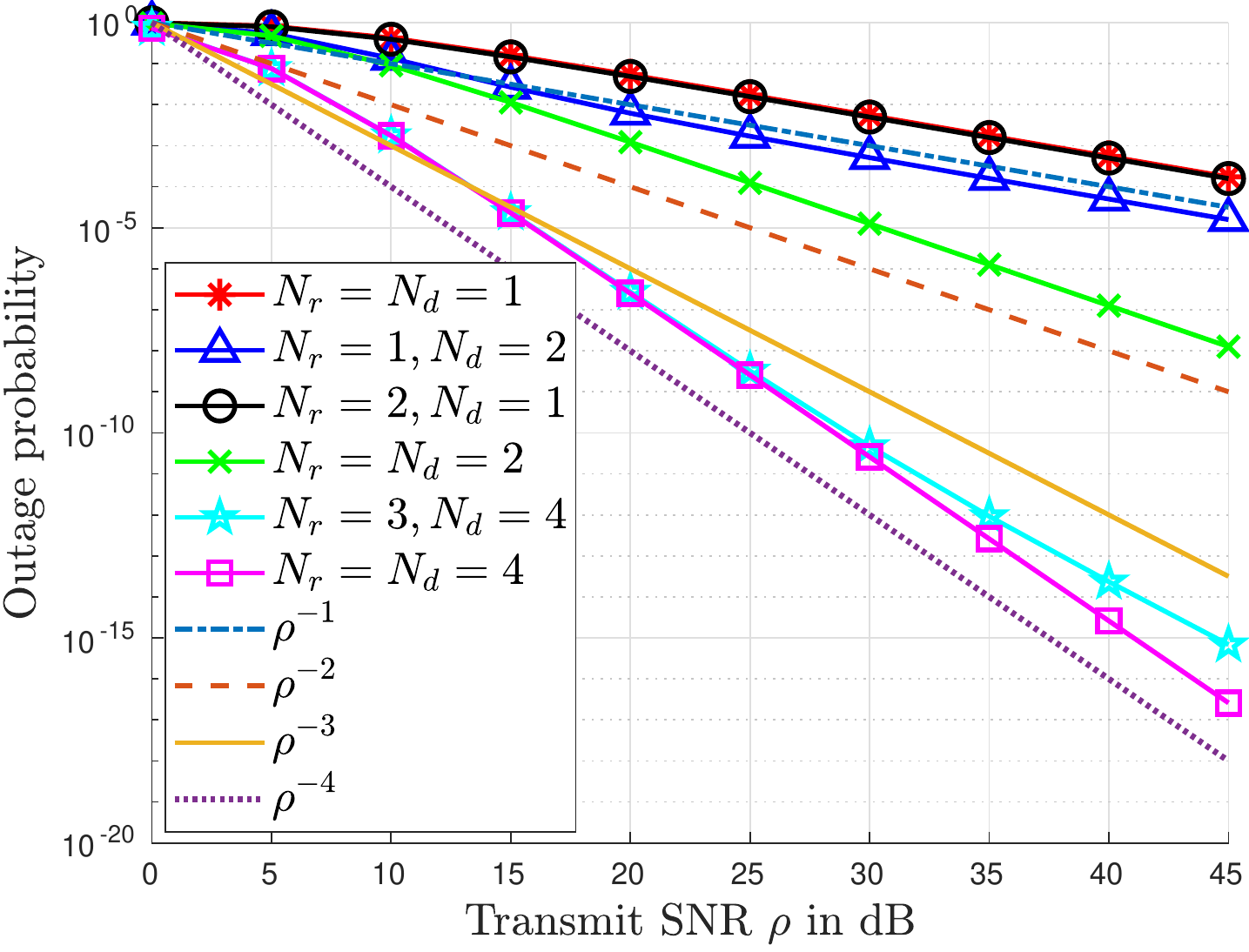}
  \caption{Symbol $s_{1}$}
  \label{Outage_s1_MRC}
\end{subfigure}%
\begin{subfigure}{.23\textwidth}
  \centering
  \includegraphics[width=1\linewidth]{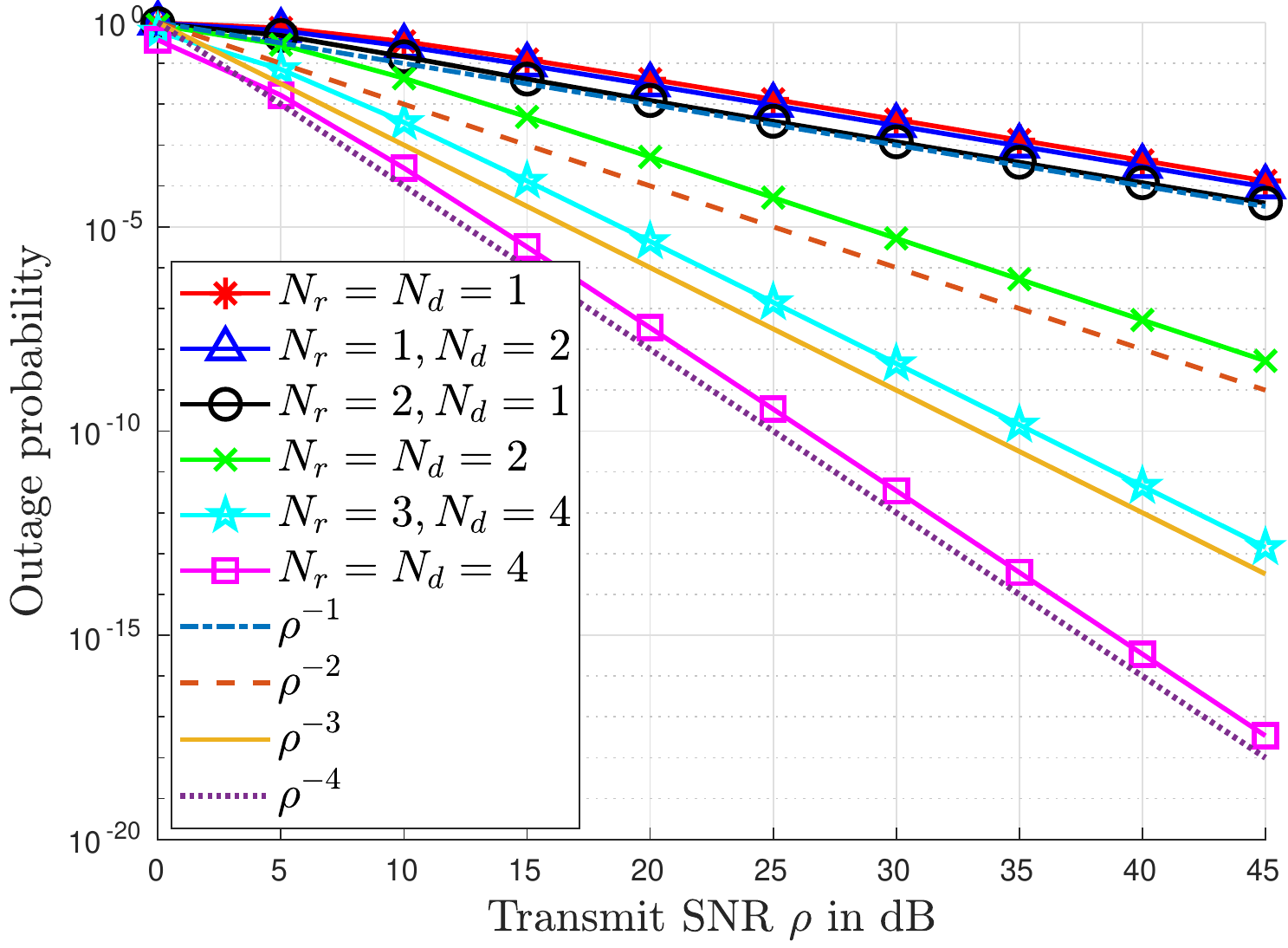}
  \caption{Symbol $s_{2}$}
  \label{Outage_s2_MRC}
\end{subfigure}
\caption{Outage probability for CRS-NOMA using MRC.}
\label{Outage_MRC}
\end{figure}

Fig. \ref{Outage_MRC} shows the outage probability of the symbols $s_{1}$ and $s_{2}$ with varying transmit SNR $\rho$ for the CRS-NOMA system using MRC. It is evident from the figure that the diversity order for both symbols is $\min(N_r, N_d)$ as proved analytically. It can also be noted that the outage probabilities for symbols $s_1$ and $s_2$ are lower for the CRS-NOMA system using MRC as compared to the corresponding probabilities for the CRS-NOMA system using SC.

\section{Conclusion}
In this paper, we provided a comprehensive achievable sum-rate analysis of a CRS-NOMA system with receive diversity. We considered two different diversity combining schemes -- SC and MRC. It was shown that the CRS-NOMA system outperforms its OMA-based counterpart by achieving higher spectral efficiency. Our analysis also confirms that the CRS-NOMA can achieve the same rate as CRS-OMA, but with a smaller number of receive antennas. We also presented the outage probability analysis of the CRS-NOMA system. Diversity analysis of the CRS-NOMA system confirms that the system achieves full diversity order of $\min(N_r, N_d)$ for both SC and MRC schemes.

\appendices 
\section{Proof of Theorem 1}

Since $|h_{sr, i}|$ is Rayleigh distributed for every $i \in \{1, 2, \ldots, N_r\}$, the CDF of $|h_{sr, i^*}|$ is given by
\begin{align}
	F_{|h_{sr, i^*}|}(x) & = \left[1 - \exp\left( \dfrac{-x^{2}}{\Omega_{sr}}\right) \right]^{N_{r}} \notag \\
	& = 1 + \sum_{k = 1}^{N_{r}}(-1)^{k} \binom{N_{r}}{k} \exp \left( \dfrac{-kx^{2}}{\Omega_{sr}}\right). \notag 
\end{align}
Therefore, the CDF of $\delta_{sr}$ can be obtained as 
\begin{align}
	F_{\delta_{sr}}(x) & = \Pr(|h_{sr, i^*}|^{2} \leq x) = \Pr(|h_{sr, i^*}| \leq \sqrt{x}) \notag \\
	& =  1 + \sum_{k = 1}^{N_{r}}(-1)^{k} \binom{N_{r}}{k} \exp \left( \dfrac{-kx}{\Omega_{sr}}\right). \label{FdeltaSR}
\end{align}
The CDF of $\delta_{sd}$ (resp. $\delta_{rd}$) can be found by replacing $\Omega_{sr}$ by $\Omega_{sd}$ (resp. $\Omega_{rd}$), while also replacing $N_r$ by $N_d$, in \eqref{FdeltaSR}. The CDF of $X = \min \{\delta_{sr}, \delta_{sd}\}$ can be found as\footnote{Given two independent random variables $\mathcal{U}$ and $\mathcal{V}$ with probability density functions (PDFs) $f_{\mathcal{U}}(x)$ and $f_{\mathcal{V}}(x)$ respectively, and CDFs $F_{\mathcal{U}}(x)$ and $F_{\mathcal{V}}(x)$ respectively, the PDF of $\mathcal{W} \triangleq \min\{\mathcal{U}, \mathcal{V}\}$ is given by $f_{\mathcal{W}}(x) = f_{\mathcal{U}}(x)[1 - F_{\mathcal{V}}(x)] + f_{\mathcal{V}}(x)[1 - F_{\mathcal{U}}(x)]$ and the CDF of $\mathcal{W}$ is given by $F_{\mathcal{W}}(x) = F_{\mathcal{U}}(x) + F_{\mathcal{V}}(x) - F_{\mathcal{U}}(x) F_{\mathcal{V}}(x)$.} 
\begin{align}
	F_{X}(x) = F_{\delta_{sr}}(x) + F_{\delta_{sd}}(x) - F_{\delta_{sr}}(x) F_{\delta_{sd}}(x). \notag 
\end{align}
Therefore, 
\begin{align}
	\!\!\!1\!-\! F_{X}(x) \!=\! & \sum_{k = 1}^{N_r} \sum_{j = 1}^{N_d} (-1)^{k + j} \binom{N_r}{k} \!\!\binom{N_d}{j} \exp \left( -\chi_{k, j} x\right), \label{1-FX}
\end{align}
where $\chi_{k, j} = (k/\Omega_{sr}) + (j/\Omega_{sd})$. Using \eqref{C_s1_SC_integral} and \eqref{1-FX}, we have 
\begin{align}
	\!\!\!I_1 = & \rho \sum_{k = 1}^{N_r} \sum_{j = 1}^{N_d} (-1)^{k + j} \binom{N_r}{k} \binom{N_d}{j} \int_{0}^{\infty} \dfrac{\exp (-\chi_{k, j} x)}{1 + \rho x}\, dx \notag \\
	= & \sum_{k = 1}^{N_r} \sum_{j = 1}^{N_d} (-1)^{k + j} \binom{N_r}{k} \!\!\binom{N_d}{j} \!\!\exp \!\!\left(\! \dfrac{\chi_{k, j}}{\rho}\!\right) \!\Gamma \left(\!\! 0, \!\dfrac{\chi_{k, j}}{\rho}\!\!\right), \label{I1}
\end{align}
where the integral above is solved using \cite[eqn.~(3.352-4),~p.~341]{Grad} and the fact that $-\operatorname{Ei}(-x) = \Gamma(0, x)$. Here $\operatorname{Ei}(\cdot)$ denotes the exponential integral. Similarly, using \eqref{C_s1_SC_integral} and \eqref{1-FX}, we have 
\begin{align}
	I_2 \!\!= \!\!\sum_{k = 1}^{N_r} \!\sum_{j = 1}^{N_d} \!(-1)^{k + j} \!\binom{N_r}{k} \!\! \binom{N_d}{j} \!\exp \!\left( \!\!\dfrac{\chi_{k, j}}{\rho a_2}\!\!\right) \!\Gamma \!\left(\!\! 0, \dfrac{\chi_{k, j}}{\rho a_2}\!\!\right). \label{I2}
\end{align}
Using \eqref{C_s1_SC_integral}, \eqref{I1} and \eqref{I2}, the closed-form expression for the average achievable rate for symbol $s_1$ in Rayleigh fading using SC in CRS-NOMA reduces to \eqref{C_s1_SC_Closed}; this completes the proof.

\section{Proof of Theorem 2}
Since $|h_{sr,i}|\, (1 \le i \le N_r)$ and $|h_{sd,i}|\, (1 \le i \le N_d)$ are Rayleigh distributed, the random variables $\lambda_{sr}$ and $\lambda_{sd}$ are Gamma distributed with shape $N_r$ and $N_d$ respectively, and scale $\Omega_{sr}$ and $\Omega_{sd}$ respectively. Since the shape parameters are positive integers, the corresponding CDFs for $\lambda_{sr}$ and $\lambda_{sd}$ can each be represented as a special case of the Erlang distribution. It follows that the CDF of $\mathcal{X}$ can be written as 
\begin{align}
	F_{\mathcal{X}}(x) = 1 - \exp(x \phi) \sum_{i = 0}^{N_r - 1} \sum_{j = 0}^{N_d - 1} \dfrac{x^{i + j}}{i! j! \Omega_{sr}^{i} \Omega_{sd}^{j}}, \notag 
\end{align}
where $\phi = \Omega_{sr}^{-1} + \Omega_{sd}^{-1}$. Hence $I_{3}$ in \eqref{C_s1_MRC_Integral} can be solved using~\cite[eqn.~(3.383-10), p. 348]{Grad} as 
\begin{align}
	I_{3} & = \!\!\sum_{i = 0}^{N_{r} - 1} \sum_{j = 0}^{N_{d} - 1} \dfrac{1}{i! \,j!\, \Omega_{sr}^{i} \, \Omega_{sd}^{j}}   \int_{0}^{\infty} \dfrac{\exp (-x \phi)  x^{(i + j)}}{1 + x \rho} dx \notag \\
	& = \!\!\sum_{i = 0}^{N_{r} - 1} \sum_{j = 0}^{N_{d} - 1} \dfrac{  \exp \left(\tfrac{\phi}{\rho}\right)  \Gamma(1 + i + j) }{i! \,j!\, \Omega_{sr}^{i} \, \Omega_{sd}^{j}  \rho^{(1 + i + j)}}\Gamma\left(-i - j, \dfrac{\phi}{\rho}\right). \label{I3_Closed}
\end{align}
Similarly, $I_{4}$ can be solved as 
\begin{align}
	\!\!\!I_{4} & = \!\!\!\sum_{i = 0}^{N_{r} - 1} \!\sum_{j = 0}^{N_{d} - 1} \!\!\!\dfrac{ \exp \left(\tfrac{\phi}{\rho a_{2}}\right)  \Gamma(1 + i + j) }{i! \,j!\, \Omega_{sr}^{i} \, \Omega_{sd}^{j} (\rho a_{2})^{(1 + i + j)}} \,\Gamma \!\left(\!-i \!- j, \dfrac{\phi}{\rho a_{2}}\!\right).\!\label{I4_Closed}
\end{align}
Using \eqref{C_s1_MRC_Integral}, \eqref{I3_Closed} and \eqref{I4_Closed}, the closed-form expression for the average achievable rate of symbol $s_{1}$ for CRS-NOMA using MRC in Rayleigh fading reduces to \eqref{C_s1_MRC_Closed}; this completes the proof.

\section*{Acknowledgment}
This publication has emanated from research conducted with the financial support of Science Foundation Ireland (SFI) and is co-funded under the European Regional Development Fund under Grant Number 13/RC/2077.

\bibliographystyle{IEEEtran}
\bibliography{SCC2019}
\end{document}